\documentclass[11pt,a4paper]{article}
\usepackage{amsmath,amsthm,amsfonts,amssymb}
\usepackage{mathptmx}
\usepackage{mathrsfs}
\usepackage{fullpage}

\newcommand{\bra}[1]{\langle #1 |}
\newcommand{\ket}[1]{| #1 \rangle}

\newcommand{\ceil}[1]{\lceil #1 \rceil}

\newcommand{\myvec}[1]{\mathbf #1}

\newcommand{\QQ}{\mathsf{Q}}
\newcommand{\RR}{\mathsf{R}}

\newtheorem{theorem}{Theorem}

\begin{document}

\author{
Fran{\c c}ois Le Gall\\ 
\texttt{legall@is.s.u-tokyo.ac.jp} \\
     Department of Computer Science\\
     Graduate School of Information Science and Technology\\
    The University of Tokyo
}
\title{Quantum Private Information Retrieval with Sublinear Communication Complexity}
\date{}
\maketitle

\begin{abstract}
This note presents a quantum protocol for private information retrieval, in the single-server case and with information-theoretical privacy, 
that has $O(\sqrt{n})$-qubit communication complexity, where~$n$ denotes the size of the database. 
In comparison, it is known that any classical protocol must use $\Omega(n)$ bits of communication in this setting. 
\end{abstract}

\section{Introduction}
Private information retrieval deals with the
design and the analysis of protocols that allow a user to retrieve an item from a server without revealing which item it is retrieving. 
This field, introduced in a seminal paper by 
Chor, Kushilevitz, Goldreich, and Sudan \cite{Chor+JACM98},  has been the subject of intensive research
due to the growing ubiquity of public databases.
Examples of applications include ensuring consumer privacy in e-commerce transactions or reading webpages 
on the Internet without revealing the user's preferences. 

In the case of a single server and of information-theoretical privacy, which is the focus of this note, private information retrieval can be described as follows.
The server has a database $\myvec{A}=(\myvec{a}^1,\myvec{a}^2,\cdots,\myvec{a}^{\ell})\in\Sigma^\ell$, 
where $\Sigma=\{0,1\}^r$ is a set of items represented as $r$-bit strings, 
and the user has an index $i\in\{1,\ldots,\ell\}$. 
A private information retrieval protocol is a (classical or quantum) communication protocol between the server and the user such that, 
when the user and the server both follow the protocol,
the user always outputs the item $\myvec{a}^i$ and
the server gets no information about the index $i$, in the following sense.
Let $V_S(\myvec{A},i)$ denote the server's view of the communication generated by the protocol when the 
server has input $\myvec{A}$ and the user has input $i$.
The privacy condition is that,
for any database $\myvec{A}\in\Sigma^\ell$ and any two indexes $i,j\in\{1,\ldots,\ell\}$, the views $V_S(\myvec{A},i)$ and 
$V_S(\myvec{A},j)$ are identical. Note that, while several subtleties arise when trying to formally define the server's view in an 
arbitrary quantum protocol, the above description will be sufficient for our purpose due to the limited interaction between the 
server and the user in the quantum protocols described in this note.


It is easy to show that, classically, downloading the whole database is essentially optimal:
any classical protocol must communicate a number of bits linear in the size of the database  \cite{Chor+JACM98}.
The communication complexity of quantum protocols for private information retrieval has first 
been investigated by Kerenidis and de Wolf \cite{Kerenidis+JCSS04}. Their work focused on 
two-message quantum protocols, and established a connection with locally 
decodable codes and random access codes. In particular it was proved that, for a single server,
any private two-message quantum protocol must use a linear amount of communication.
This note shows that this lower bound does not hold for quantum protocols using more than two messages 
and describes how to construct a three-message quantum protocol for private information retrieval
with sublinear communication complexity, thus breaking for the first time the linear barrier in the 
single-server and information-theoretical privacy setting. 
Our main result is the following theorem.
\begin{theorem}\label{the1}
Let $\ell$ and $r$ be any positive integers.
There exists a private information retrieval quantum protocol 
that, for any database $\myvec{A}\in\Sigma^\ell$ with $\Sigma=\{0,1\}^r$, uses $2\ell+2r$ qubits of communication.
\end{theorem}
 
Since the overall size of the 
database is $\ell r$ bits, Theorem \ref{the1} gives a quadratic improvement over classical protocols and two-message 
quantum protocols
whenever $\ell+r=O(\sqrt{\ell r})$, for example when $\ell=\Theta(r)$.
This quadratic improvement can actually be obtained for any values of 
$\ell$ and $r$: the idea is to decompose the database into about $\sqrt{\ell r}$ blocks, each of size about $\sqrt{\ell r}$ bits.
To illustrate this, let us consider a binary database $\myvec{A}=(a^1,\ldots,a^\ell)$ when
$\ell=s^2$ for some positive integer $s$. 
We construct the database $\myvec{B}=(\myvec{b^{1}}, \ldots, \myvec{b^{s}})$ such that,
for each $k\in\{1,\ldots,s\}$, the $k$-th block is $\myvec{b^k}=(a^{(k-1)s+1},\ldots, a^{ks})\in\{0,1\}^s$.
Note that the bit $a^i$ is contained in the block $\myvec{b^{j}}$ with $j=\ceil{i/s}$. By running the protocol
of Theorem \ref{the1} where, as inputs, the server has database $\myvec{B}$ and the user has index $j$, the user
is able to recover the whole block $\myvec{b^{j}}$, and thus the bit $a^i$, using $O(s)$ qubits of communication. 

We stress that this note considers only the setting where the parties do not deviate from the protocol,
as often assumed in works focusing on algorithmic or complexity-theoretic aspects of private information retrieval.
While this restriction may reduce the applicability of our result, we believe that it 
nevertheless illustrates the subtle interplay of interaction and quantum information in protecting privacy.
Indeed, even in this setting, a linear amount of communication is needed for classical
protocols and for two-message quantum protocols. 

{\bf Other related works.}
Several other aspects of quantum protocols for private information retrieval have been investigated. 
The case of multiple servers has been studied in \cite{Kerenidis+JCSS04,Kerenidis+IPL04}, while the case of 
symmetric private information retrieval, where the server's privacy is also taken into consideration, has been
studied in \cite{Kerenidis+IPL04,Giovannetti+PRL08,Jain+JACM09}. 
Privacy issues in quantum communication complexity have been studied in \cite{Klauck04} as well. 
Let us mention that quantum protocols for symmetric private information retrieval are also studied under the name 
of quantum oblivious transfer protocols, especially when the server and the user may deviate from the protocol (i.e., when considering malicious parties).  

\section{Proof of Theorem~\ref{the1}}

We suppose that the reader is familiar with quantum computation and refer to, e.g., \cite{Nielsen+00} for an introduction to this field.
Let us first describe some of our notations. 
Given two bits $a,b\in\{0,1\}$, we write their parity as $a\oplus b$.
For any two elements $\myvec{u}=(u_1,\ldots,u_r)$ and $\myvec{v}=(v_1,\ldots,v_r)$ in $\Sigma=\{0,1\}^r$, 
let us write $\myvec{u}\cdot \myvec{v}=u_1v_1\oplus\cdots\oplus u_rv_r$
and $\myvec{u}\oplus \myvec{v}=(u_1\oplus v_1,\ldots,u_r\oplus v_r)$. Note that $\myvec{u}\cdot \myvec{v}$ is a bit
and $\myvec{u}\oplus \myvec{v}$ is an element of $\Sigma$. 
Our protocol will use the Pauli gate
$${\rm Z}:=\sum_{z\in\{0,1\}}(-1)^z\ket{z}\bra{z}$$ acting on one qubit and the Quantum Fourier Transform
$${\rm QFT}:=\frac{1}{\sqrt{|\Sigma|}}\sum_{\myvec{y},\myvec{z}\in\Sigma}(-1)^{\myvec{y}\cdot \myvec{z}}\ket{\myvec{y}}\bra{\myvec{z}}$$
acting on $r$ qubits.
It will also use the gates
\begin{eqnarray*}
{\rm CNOT}^{(\RR_1,\RR_2)}&:=&\sum_{\myvec{y},\myvec{z}\in\Sigma}\ket{\myvec{y}}_{\RR_1}\ket{\myvec{z}\oplus \myvec{y}}_{\RR_2}\bra{\myvec{y}}_{\RR_1}\bra{\myvec{z}}_{\RR_2}\\
{\rm U}^{(\RR_1,\QQ)}_\myvec{b}&:=&\sum_{\myvec{y}\in\Sigma,z\in\{0,1\}}\ket{\myvec{y}}_{\RR_1}\ket{z\oplus \myvec{b}\cdot \myvec{y}}_{\QQ}\bra{\myvec{y}}_{\RR_1}\bra{z}_{\QQ}, 
\end{eqnarray*}
where $\RR_1$ and $\RR_2$ denote $r$-qubit  registers, $\QQ$ denotes a one-qubit  register, and $\myvec{b}$ is any element in $\Sigma$.

We now present the proof of Theorem~\ref{the1}.


\begin{proof}[Proof of Theorem~\ref{the1}]
The protocol uses $\ell+2$ quantum registers:
Registers $\mathsf{R}$ and $\mathsf{R'}$ each consisting of $r$ qubits, and Registers 
$\mathsf{Q_1},\ldots,\mathsf{Q_\ell}$ each consisting of one qubit.
For any database $\myvec{A}=(\myvec{a}^1,\ldots,\myvec{a}^\ell)\in\Sigma^\ell$, let us 
denote by $\ket{\Phi_\myvec{A}}$ the quantum state 
\[
\ket{\Phi_\myvec{A}}:=\frac{1}{\sqrt{2^r}}\sum_{\myvec{x}\in \Sigma}
\ket{\myvec{x}}_{\mathsf{R}}\ket{\myvec{x}}_{\mathsf{R'}}\ket{\myvec{x}\cdot \myvec{a^{1}}}_{\mathsf{Q_1}}\cdots\ket{\myvec{x}\cdot \myvec{a^{\ell}}}_{\mathsf{Q_\ell}}
\]
in Registers $(\mathsf{R},\mathsf{R'},\mathsf{Q_1},\ldots,\mathsf{Q_\ell})$.
The protocol is described in Figure~\ref{figure:QPIR}. It consists of three messages and uses 
a total amount of $2\ell+2r$ qubits of communication.
\begin{figure}[h!]
\hrule\hspace{2mm}
\begin{enumerate}
\item[]
{\bf Server's input:} $\myvec{A}=(\myvec{a}^1,\ldots,\myvec{a}^\ell)\in\Sigma^\ell$
\item[]
{\bf User's input:} $i\in\{1,\ldots,\ell\}$
\item
The server constructs the quantum state $\ket{\Phi_\myvec{A}}$ and
sends Registers $\mathsf{R'}$, $\mathsf{Q_1},\ldots,\mathsf{Q_\ell}$ to the user.
\item
The user applies ${\rm Z}$ over Register $\QQ_i$ and sends back Registers $\mathsf{Q_1},\ldots,\mathsf{Q_\ell}$ to the server.
\item
The server applies ${\rm U}_{\myvec{a}^k}^{(\RR,\QQ_k)}$, for each $k\in\{1,\ldots,\ell\}$, and sends to the user Register $\RR$.
\item
The user applies ${\rm CNOT}^{(\RR,\RR')}$, applies ${\rm QFT}$ over Register ${\mathsf{R}}$, and then
measures $\mathsf{R}$ in the computational basis. 
\end{enumerate} 
\hrule
\caption{Quantum private information retrieval protocol.}
\label{figure:QPIR}
\end{figure}

We first show that in this protocol the user always outputs the correct element of the database.
Observe that, at the end of Step 2, the state is
\[
\ket{\Phi}=\frac{1}{\sqrt{2^r}}\sum_{\myvec{x}\in \Sigma}(-1)^{\myvec{x}\cdot \myvec{a^i}}
\ket{\myvec{x}}_{\mathsf{R}}\ket{\myvec{x}}_{\mathsf{R'}}\ket{\myvec{x}\cdot \myvec{a^{1}}}_{\mathsf{Q_1}}\cdots\ket{\myvec{x}\cdot \myvec{a^{\ell}}}_{\mathsf{Q_\ell}}.
\]
At Step 4, just before the user performs the measurement, the state is 
$\ket{\myvec{a^i}}_{\mathsf{R}}\ket{\myvec{0}}_{\mathsf{R'}}\ket{0}_{\mathsf{Q_1}}\cdots\ket{0}_{\mathsf{Q_\ell}}$,
and measuring Register $\RR$ gives the element $\myvec{a^i}$ with probability 1.
Let us now consider the user's privacy. 
The only information about $i$ that a server following the protocol can obtain is from 
Registers $\mathsf{R},\mathsf{Q}_1,\ldots,\mathsf{Q}_\ell$ of the state $\ket{\Phi}$.
Since tracing out Register $\RR'$ in $\ket{\Phi}\bra{\Phi}$ gives
the density matrix
\[
\frac{1}{2^r}\sum_{\myvec{x}\in \Sigma}\ket{\myvec{x}}_{\mathsf{R}}\ket{\myvec{x}\cdot \myvec{a^{1}}}_{\mathsf{Q_1}}
\cdots\ket{\myvec{x}\cdot \myvec{a}^{\ell}}_{\mathsf{Q_\ell}}
\bra{\myvec{x}}_{\mathsf{R}} \bra{\myvec{x}\cdot \myvec{a^{1}}}_{\mathsf{Q_1}} \cdots\bra{\myvec{x}\cdot \myvec{a^{\ell}}}_{\mathsf{Q_\ell}},
\]
the server obtains no information about the user's input.
\end{proof}

\noindent{\bf Remark.} 
As already mentioned, in this note we only consider the case where the server follows the protocol.
This assumption is used in the analysis of the protocol of Figure~\ref{figure:QPIR} in order to ensure that the server prepares the 
state $\ket{\Phi_\myvec{A}}$ at Step 1.
Note that if, instead of $\ket{\Phi_\myvec{A}}$, the server prepared for example the state
\[
\ket{\Phi'_\myvec{A}}:=\frac{1}{\sqrt{2^r}}\sum_{\myvec{x}\in \Sigma}
\ket{\myvec{x}}_{\mathsf{R}}\ket{\myvec{0}}_{\mathsf{R'}}\ket{\myvec{x}\cdot \myvec{a^{1}}}_{\mathsf{Q_1}}\cdots\ket{\myvec{x}\cdot \myvec{a^{\ell}}}_{\mathsf{Q_\ell}},
\]
then it would be able to recover the index $i$ with probability one at Step 3. 

%

\section*{Acknowledgements}
The author is grateful to Takeshi Koshiba, Harumichi Nishimura, and Ronald de Wolf  for helpful discussions
about this work. He also
acknowledges support from the JSPS,
under the grant-in-aid for research activity start-up No.~22800006.


\end{document}